\documentclass[12pt,journal,draftcls,letterpaper,onecolumn]{IEEEtran}

\usepackage{cite}      

\usepackage{graphicx}  

\usepackage{subfigure} 

\usepackage{amsmath}   

\newtheorem{defi}{Definition}[section]
\newtheorem{thm}{Theorem}[section]
\newtheorem{lem}[thm]{Lemma}

\begin{document}
%
\title{Hamiltonian Connectivity of Twisted Hypercube-Like Networks under the Large Fault Model}
%
%

\author{Qiang~Dong,~Hui~Gao,~Yan~Fu~and~Xiaofan~Yang~
\thanks{}
\thanks{Q. Dong, H. Gao and Y. Fu are with the School of Computer Science and Engineering,
University of Electronic Science and Technology of China, Chengdu
611731, China. The first two authors contributed equally to the
work.}
\thanks{X. Yang is with the College of Computer Science, Chongqing University,
Chongqing 400044, China.}
\thanks{Corresponding Author: Q. Dong (E-mail: dongq@uestc.edu.cn).}}

%
%
%
\markboth{IEEE TPDS,~Vol.~xx, No.~xx,~October~2011}{Dong
\MakeLowercase{\textit{et al.}}: Hamiltonian Connectivity of THLNs
under the Large Fault Model}



\maketitle

\begin{abstract}
Twisted hypercube-like networks (THLNs) are an important class of
interconnection networks for parallel computing systems, which
include most popular variants of the hypercubes, such as crossed
cubes, M\"obius cubes, twisted cubes and locally twisted cubes. This
paper deals with the fault-tolerant hamiltonian connectivity of
THLNs under the large fault model. Let $G$ be an $n$-dimensional
THLN and $F \subseteq V(G)\bigcup E(G)$, where $n \geq 7$ and $|F|
\leq 2n - 10$. We prove that for any two nodes $u,v \in V(G - F)$
satisfying a simple necessary condition on neighbors of $u$ and $v$,
there exists a hamiltonian or near-hamiltonian path between $u$ and
$v$ in $G-F$. The result extends further the fault-tolerant graph
embedding capability of THLNs.
\end{abstract}


\section{Introduction}
%
%
%
%



\PARstart{T}{he} performance of a parallel computing system heavily
depends on the effectiveness of the underlying interconnection
network. An interconnection network is usually represented by a
graph, where nodes and edges correspond to processors and
communication links between processors, respectively. In the design
and analysis of an interconnection network, one major concern is its
graph embedding capability, which reflects how efficiently a
parallel algorithm with structured task graph (guest graph) can be
executed on this network (host graph). Cycles and paths are
recognized as important guest graphs because a great number of
parallel algorithms, such as matrix-vector multiplication, Gaussian
elimination and bitonic sorting, have been developed on
cycle/path-structured task graphs \cite{P99}.

As the size of a parallel computing system increases, it becomes
much likely that some processors and communication links fail to
work in such a system. Consequently, it is essential to study the
fault-tolerant graph embedding capability of an interconnection
network with faulty elements.

The \textit{hypercube-like networks} (HLNs) are an important class
of generalizations of the popular hypercube interconnection networks
for parallel computing. Among HLNs one may identify a subclass of
networks, called the \textit{twisted hypercube-like networks}
(THLNs), which include most well-known variants of the hypercubes,
such as crossed cubes \cite{EFE92}, M\"obius cubes \cite{CULL95},
twisted cubes \cite{HK87} and locally twisted cubes \cite{YXF05}.
The fault-free and fault-tolerant cycle/path embedding capabilities
of these hypercube variants have been intensively studied in the
literature
\cite{DONG08,DONG08IPL,FAN10,HF10,H04,H10,LAI11,WF11,XU11,XU06,YMC10,YXF10}.

In recent years, the fault-tolerant cycle/path embedding
capabilities of HLNs and THLNs have received considerable research
attention \cite{FAN11TCS,FAN11,H10SIAM,PARK05,PARK09,PARK07,YANG11}.
However, most of the embeddings tolerate no more faulty elements
than the degree of the graph, i.e., under the \emph{small fault
model}. Recently, Yang \emph{et al.} \cite{YANG11} studied the cycle
embedding capability of THLNs with more faulty elements than the
degree of the graph, i.e., under the \emph{large fault model}. They
proved that for an $n$-dimensional ($n$-D) THLN $G$ and $F \subseteq
V(G) \bigcup E(G)$, where $n \ge 7$ and $|F| \le 2n - 9$, $G - F$
contains a hamiltonian cycle if $\delta (G - F) \ge 2$, and $G - F$
contains a near-hamiltonian cycle if $\delta (G - F) \le 1$.

A question arises naturally: what about the fault-tolerant
hamiltonian paths in THLNs under the large fault model? This paper
attempts to partially answer this question. Let $G$ be an $n$-D THLN
and $F \subseteq V(G)\bigcup E(G)$, where $n \geq 7$ and $|F| \leq
2n - 10$. We prove that for any two nodes $u,v \in V(G - F)$
satisfying a simple necessary condition on neighbors of $u$ and $v$,
there exists a hamiltonian or near-hamiltonian path between $u$ and
$v$ in $G-F$. As a nontrivial extension of \cite{YANG11}, our result
extends further the fault-tolerant graph embedding capability of
THLNs.

The rest of this paper is organized as follows. Section 2 gives
definitions and notions. Section 3 establishes the main result.
Section 4 concludes the paper.

\section{Definitions and Notations}
\label{}

For basic graph-theoretic notations and terminology, the reader is
referred to ref. \cite{DIESTEL}. For a graph $G$, let $V(G)$ and
$E(G)$ denote its node set and edge set, respectively. For two nodes
$u$ and $v$ in a graph $G$, $u$ is a neighbor of $v$ if and only if
$(u, v)\in E(G)$. For a node $u$ in a graph $G$, let $N_G (u) =
\big\{ {v \in V(G):( u,v)  \in E(G)} \big\}$, and the degree of $u$
in $G$ is defined as $\deg_G(u) = \big| {N_G (u)} \big|$. If the
degree of every node in a graph $G$ is $k$, then $G$ is called a
$k$-regular graph. For a graph $G$, let $\delta(G)=\mathop {\min}
\limits_{u \in V(G)}\big\{ \deg_G(u)\big\}$. For a graph $G$ and a
set $F \subseteq V(G) \bigcup E(G)$, let $G - F$ denote the graph
defined by $V(G - F) = V(G) - F$, $E(G - F) = \big\{ (u,v)\in E(G):
u,v \in V(G) - F {\rm{\ and\ }} (u, v) \notin F \big\}$.

A hamiltonian cycle (hamiltonian path, respectively) in a graph is a
cycle (path, respectively) that passes every node of the graph
exactly once. A near-hamiltonian cycle (near-hamiltonian path,
respectively) in a graph is a cycle (path, respectively) that passes
every node but one of the graph exactly once.

For two nodes $u$ and $v$ in a graph $G$, let $dist_{G}(u, v)$
denote the distance between $u$ and $v$, i.e., the minimum length of
all paths between $u$ and $v$. For a node $x$ on a path $P$ between
$u$ and $v$, if $dist_{P}(x,u)\leq dist_{P}(x,v)$, then we regard
$x$ as a $u$-closer node on $P$, and vice-versa.

According to \cite{YANG11}, we give the definition of twisted
hypercube-like network as follows.

\begin{figure}
  \begin{minipage}[t]{0.5\linewidth}
    \includegraphics[scale=0.75]{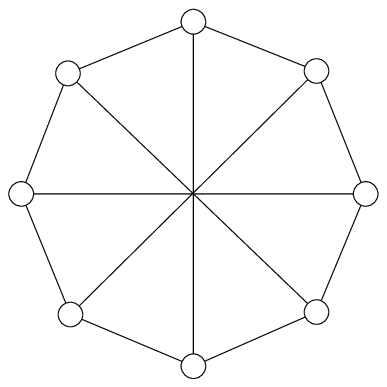}
    \caption{3-D THLN}
    \label{fig:1}
  \end{minipage}%
  \begin{minipage}[t]{0.5\linewidth}
    \includegraphics[scale=0.75]{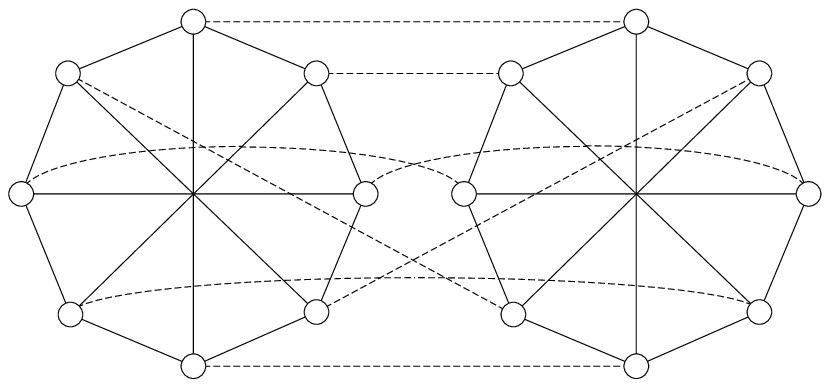}
    \caption{4-D crossed cube}
    \label{fig:2}
  \end{minipage}
\end{figure}

\begin{figure}
  \begin{minipage}[t]{0.5\linewidth}
    \includegraphics[scale=0.75]{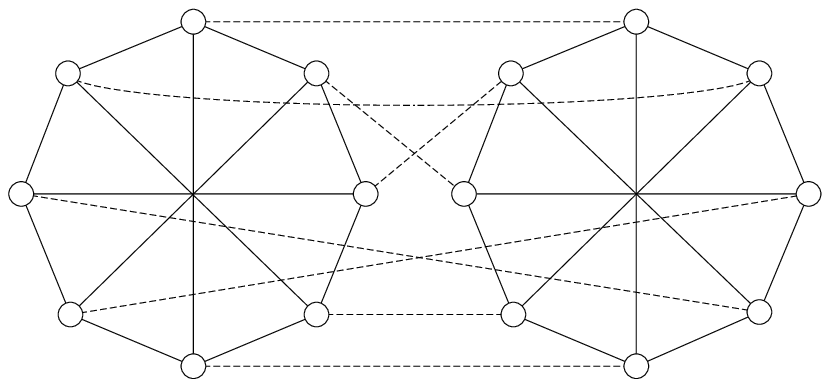}
    \caption{4-D 0-M\"obius cube}
    \label{fig:3}
  \end{minipage}%
  \begin{minipage}[t]{0.5\linewidth}
    \includegraphics[scale=0.75]{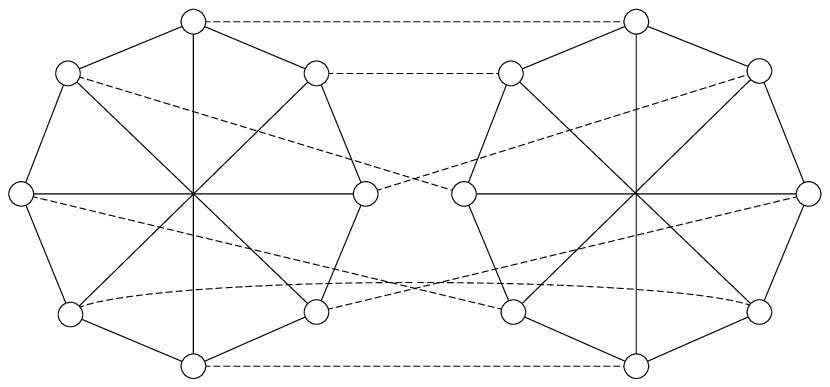}
    \caption{4-D locally twisted cube}
    \label{fig:4}
  \end{minipage}
\end{figure}

\begin{defi}
\label{defi:THLN}  For $n \ge 3$, an $n$-dimensional ($n$-D, for
short) \textit{twisted hypercube-like network} (THLN, for short) is
a graph $G$ defined recursively as follows.

\begin{itemize}[\setlabelwidth{(1)}]
\item[(1)] For $n = 3$, $G$ is isomorphic to the graph in
Fig. \ref{fig:1}.
\item[(2)] For $n \ge 4$, $G$ is constructed from two $(n - 1)$-D THLN copies,
$G_{1}$ and $G_{2}$, in this way:

\centerline{$V(G) = V(G_1 ) \bigcup V(G_2 )$,}

\centerline{$E(G) = E(G_1 ) \bigcup E(G_2 ) \bigcup \big \{ \big
(u,\phi (u)\big ) :u \in V(G_1 ) \big\}$,}

where $\phi :V(G_1 ) \to V(G_2 )$ is a bijective mapping. In what
follows, we denote such a THLN as $G =  \oplus _\phi (G_1, G_2 )$,
and we use $E_c$ to denote the edge set $\big \{ \big (u,\phi
(u)\big ) :u \in V(G_1 ) \big\}$.

\end{itemize}
\end{defi}

Fig. \ref{fig:2}--\ref{fig:4} illustrate three well-known subclasses
of 4-D THLNs. It is easily veritied that, an $n$-D THLN $G$ is an
$n$-regular graph, $\big|V(G)\big|=2^n$, $\big|E(G)\big|=2^{n-1}\ast
n$, and $\big|E_c\big|=2^{n-1}$. The following important results on
THLNs reported in \cite{PARK05,YANG11,PARK09} will be used in this
paper.

\begin{lem}
\label{lem:hamilton} Let $G$ be an $n$-D THLN and $F \subseteq V(G)
\bigcup E(G)$, where $n \geq 3$. $G - F$ contains a hamiltonian
cycle if $\vert F\vert \le n - 2$, and $G - F$ contains a
hamiltonian path between any two fault-free nodes if $\vert F\vert
\le n - 3$.
\end{lem}

\begin{lem}
\label{lem:YANG} Let $G$ be an $n$-D THLN and $F \subseteq V(G)
\bigcup E(G)$, where $n \geq 7$ and $|F| \leq 2n - 9$. $G - F$
contains a hamiltonian cycle if $\delta (G - F) \geq 2$, and $G - F$
contains a near-hamiltonian cycle if $\delta (G - F) \leq 1$.
\end{lem}

\begin{lem}
\label{lem:disjoint} Let $G$ be an $n$-D THLN and $F \subseteq V(G)
\bigcup E(G)$ such that $|F|\leq n - 4$. For any two pairs of nodes
$[x_1,x_2]$ and $[y_1, y_2]$ in $G - F$, there exist two paths $P_1$
and $P_2$ in $G - F$ such that $P_1$ connects $x_1$ and $y_1$, $P_2$
connects $x_2$ and $y_2$, $V(P_1) \bigcap V(P_2) = \emptyset$ and
$V(P_1) \bigcup V(P_2) = V (G - F) $.
\end{lem}

\section{Main result}
\label{}

This section deals with the fault-tolerant hamiltonian connectivity
of THLNs under the large fault model. We can easily verify the
following lemma.

\begin{lem}
Given a graph $G$ and $F \subseteq V(G)\bigcup E(G)$. For any two
nodes $x,y\in V(G-F)$ such that $N_{G-F}(s)-\{t\}= \emptyset$ or
$N_{G-F}(t)-\{s\}= \emptyset$, there exists no path of length two or
longer in $G-F$.
\end{lem}

Excluding the above special cases, the main result of this paper is
formulated as follows.

\begin{thm}
Let $G$ be an $n$-D THLN and $F \subseteq V(G)\bigcup E(G)$, where
$n \geq 7$ and $|F| \leq 2n - 10$. For any two nodes $s,t \in V(G -
F)$ such that $N_{G-F}(s)-\{t\}\neq \emptyset$ and
$N_{G-F}(t)-\{s\}\neq \emptyset$, there exists a hamiltonian or
near-hamiltonian path between $s$ and $t$ in $G-F$.
\end{thm}

\begin{proof}
We argue the assertion by induction on $n$. Let $G$ be a 7-D THLN
and $F \subseteq V(G) \bigcup E(G)$, where $|F| \leq 2\times 7 - 10
= 4$. By Lemma \ref{lem:hamilton}, $G - F$ is hamiltonian connected.
Hence, the assertion is true for $n = 7$.

Suppose the assertion holds for $n=k \ge 7$. Let $G = \oplus _{\phi}
(G_1, G_2)$ be a $(k + 1)$-D THLN, where $G_1$ and $G_2$ are $k$-D
THLNs, $E_c  = \big\{ \big( u,\phi (u)\big) :u \in V(G_1 ) \big\}$.
In the following discussion, we use a lowercase with subscript 1 to
denote a node in $V(G_1)$, and the same lowercase with subscript 2
to denote the node in $V(G_2)$ such that these two nodes are
connected by an edge in $E_c$. For example, $x_1\in V(G_1)$, $x_2\in
V(G_2)$, and $(x_1,x_2)\in E_c$.

Let $F \subseteq V(G) \bigcup E(G)$, where $\left| F \right| \leq
2(k + 1) - 10 = 2k - 8$, and let $F_1  = F \bigcap \big( {V(G_1)
\bigcup E(G_1)} \big)$, $F_2 = F \bigcap \big( {V(G_2) \bigcup
E(G_2)} \big)$, and $F_c = E_c - E(G-F) $. Without loss of
generality (W.L.O.G., for short), we may assume  $\vert F_{1}\vert
\ge \vert F_{2}\vert $, then $\vert F_{2}\vert \le k - 4$. The
discussion will proceed by distinguishing the following five cases.

\bigskip \noindent \emph{Case 1.} $\vert F_{1}\vert  \le  2k - 10$.

\bigskip \noindent \emph{Case 1.1.} $s, t \in V(G_1) $.

Since $\vert F_{1}\vert  \le  2k - 10$, $N_{G_1-F_1}(s)-\{t\}=
\emptyset$ and $N_{G_1-F_1}(t)-\{s\}= \emptyset$ can not happen
simultaneously.

\bigskip \noindent \emph{Case 1.1.1.} $N_{G_1-F_1}(s)-\{t\}\neq \emptyset$ and
$N_{G_1-F_1}(t)-\{s\}\neq \emptyset$ (see Fig. \ref{fig:1.1.1}).

According to induction hypothesis, there exists a hamiltonian or
near-hamiltonian path $P_1$ between $s$ and $t$ in $G_1 - F_1$. We
claim that we can find an edge $(u_1, v_1)$ on $P_1$ such that $u_2
\in N_{G-F}(u_1)$ and $v_2 \in N_{G-F}(v_1)$. The existence of such
an edge is due to the fact that there are at least $2^k - (2k - 10)
- 2 = 2^k - 2k + 8$ candidate edges on $P_1$, and there are at most
$2k - 8 < \big\lceil(2^k - 2k + 8)/2\big\rceil$ faulty elements in
$G_2$ and $E_c$, each of which can "block" at most two candidates.

We may write $P_1$ as $\langle s,P_{11},u_1,v_1,P_{12},t\rangle$. By
Lemma \ref{lem:hamilton}, there exists a hamiltonian path $P_2$
between $u_2$ and $v_2$ in $G_2 - F_2$. Thus, $\langle
s,P_{11},u_1,u_2,P_2,v_2,$ $v_1,P_{12},t\rangle$ forms a hamiltonian
or near-hamiltonian path between $s$ and $t$ in $G - F$.

\begin{figure}
    \centering
    \includegraphics[scale=0.7]{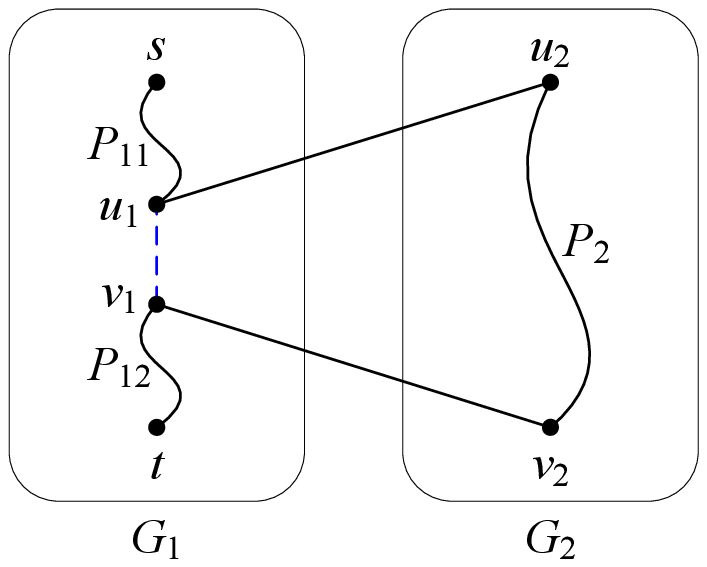}
    \caption{Illustration of the hamiltonian or near-hamiltonian path in Case 1.1.1}
    \label{fig:1.1.1}
\end{figure}

\bigskip \noindent \emph{Case 1.1.2.} Either $N_{G_1-F_1}(s)-\{t\}= \emptyset$
or $N_{G_1-F_1}(t)-\{s\}= \emptyset$ (see Fig. \ref{fig:1.1.2}).

W.L.O.G., we may assume $N_{G_1-F_1}(t)-\{s\}= \emptyset$. Let
$t_1=t$. According to the assumption $N_{G-F}(t)-\{s\}\neq
\emptyset$, we have $(t_1,t_2)\in E_c - F_c$. Since $|E_c| - |F_c|
\geq 2^k - (2k-8) \geq 122$ for $k\geq 7$, we can find an edge
$(u_1, u_2)\in E_c - F_c$ such that $u_1 \neq s$ and $u_1 \neq t$.
Clearly, $N_{G_1-F_1}(s)-\{u_1\}\neq \emptyset$ and
$N_{G_1-F_1}(u_1)-\{s\}\neq \emptyset$.

According to induction hypothesis, there exists a near-hamiltonian
path $P_1$ between $s$ and $u_1$ in $G_1 - F_1$, where $t$ is not on
$P_1$. By Lemma \ref{lem:hamilton}, there exists a hamiltonian path
$P_2$ between $u_2$ and $t_2$ in $G_2 - F_2$. Thus, $\langle
s,P_1,u_1,u_2,P_2,t_2,t\rangle$ forms a hamiltonian path between $s$
and $t$ in $G - F$.

\begin{figure}
    \centering
    \includegraphics[scale=0.7]{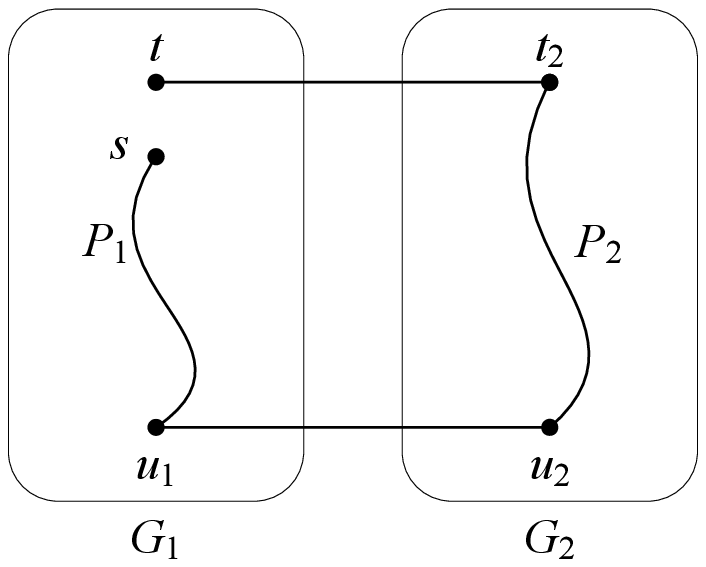}
    \caption{Illustration of the hamiltonian path in Case 1.1.2}
    \label{fig:1.1.2}
\end{figure}

\bigskip \noindent \emph{Case 1.2.} $s, t \in V(G_2) $ (see Fig. \ref{fig:1.2}).

By Lemma \ref{lem:hamilton}, there exists a hamiltonian path $P_2$
between $s$ and $t$ in $G_2 - F_2$. We can find two nonadjacent
edges $(u_2, v_2)$ and $(x_2, y_2)$ on $P_2$ such that $u_1 \in
N_{G-F}(u_2)$, $v_1 \in N_{G-F}(v_2)$, $x_1 \in N_{G-F}(x_2)$, and
$y_1 \in N_{G-F}(y_2)$. Since $|F_1|\leq 2k-10$, there exists at
most one node whose degree is less than 2 in $G_1 - F_1$. Then we
can choose one edge out of $(u_2, v_2)$ and $(x_2, y_2)$, say $(u_2,
v_2)$, such that $\deg _{G_1-F_1}(u_1)\geq 2$ and $\deg
_{G_1-F_1}(v_1)\geq 2$, which mean that
$N_{G_1-F_1}(u_1)-\{v_1\}\neq \emptyset$ and
$N_{G_1-F_1}(v_1)-\{u_1\}\neq \emptyset$.

We may write $P_2$ as $\langle s,P_{21},u_2,v_1,P_{22},t\rangle$.
According to induction hypothesis, there exists a hamiltonian or
near-hamiltonian path $P_1$ between $u_1$ and $v_1$ in $G_1 - F_1$.
Thus, $\langle s,P_{21},u_2,u_1,P_1,v_1,$ $v_2,P_{22},t\rangle$
forms a hamiltonian or near-hamiltonian path between $s$ and $t$ in
$G - F$.

\begin{figure}
    \centering
    \includegraphics[scale=0.7]{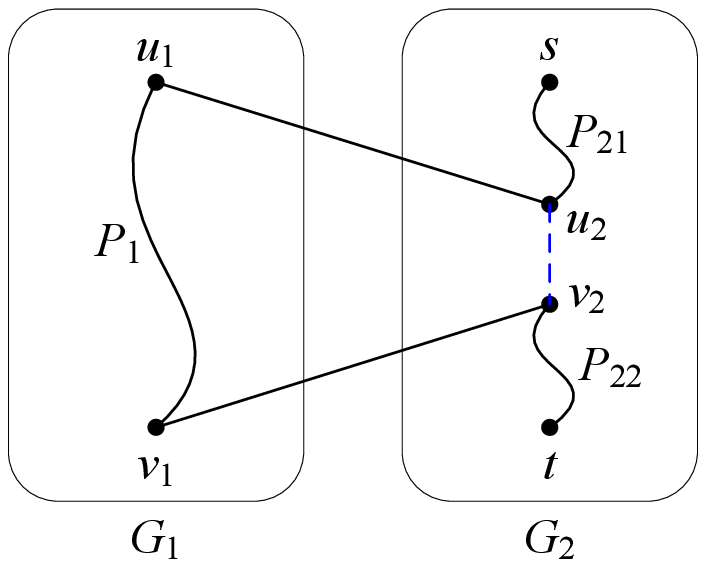}
    \caption{Illustration of the hamiltonian or near-hamiltonian path in Case 1.2}
    \label{fig:1.2}
\end{figure}

\bigskip \noindent \emph{Case 1.3.} $s\in V(G_1)$ and $t \in V(G_2)$,
or, $s\in V(G_2)$ and $t \in V(G_1)$.

W.L.O.G., we may assume $s\in V(G_1)$ and $t \in V(G_2)$ (see Fig.
\ref{fig:1.3}). First, since $|E_c| - |F_c| \geq 2^k - (2k-8) \geq
122$ for $k\geq 7$, we can find three edges $(u_1, u_2)$, $(v_1,
v_2)$ and $(w_1, w_2)$ in $E_c - F_c$ such that $u_1,v_1,w_1 \neq s$
and $u_2,v_2,w_2 \neq t$. Second, we can choose two nodes out of
$u_1$, $v_1$ and $w_1$, say $u_1$ and $v_1$, such that $\deg
_{G_1-F_1}(u_1)\geq 2$ and $\deg _{G_1-F_1}(v_1)\geq 2$. Finally,
since $N_{G-F}(s)-\{t\}\neq \emptyset$ means that
$N_{G_1-F_1}(s)\neq \emptyset$, we can choose one node out of $u_1$
and $v_1$, say $u_1$, such that $N_{G_1-F_1}(s)-\{u_1\}\neq
\emptyset$.

According to induction hypothesis, there exists a hamiltonian or
near-hamiltonian path $P_1$ between $s$ and $u_1$ in $G_1 - F_1$. By
Lemma \ref{lem:hamilton}, there exists a hamiltonian path $P_2$
between $u_2$ and $t$ in $G_2 - F_2$. Thus, $\langle
s,P_1,u_1,u_2,P_2,t\rangle$ forms a hamiltonian or near-hamiltonian
path between $s$ and $t$ in $G - F$.

\begin{figure}
    \centering
    \includegraphics[scale=0.7]{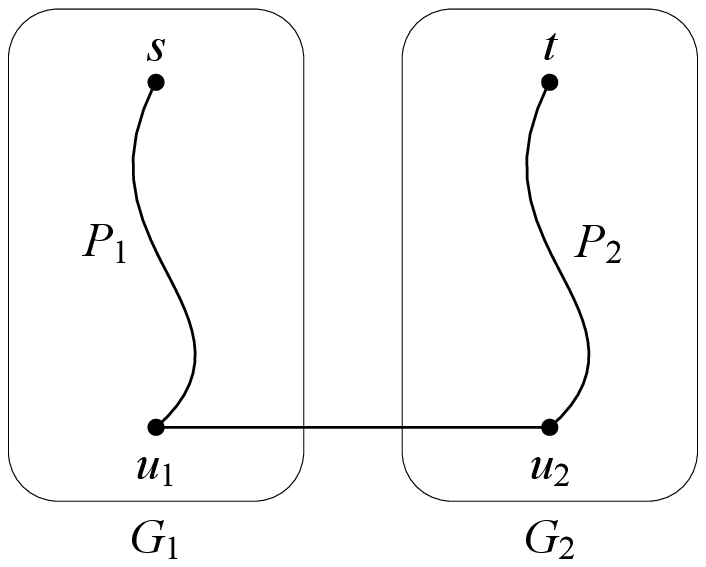}
    \caption{Illustration of the hamiltonian or near-hamiltonian path in Case 1.3}
    \label{fig:1.3}
\end{figure}

\bigskip \noindent \emph{Case 2.} $\vert F_{1}\vert  =  2k - 9$ and $\delta (G_1 - F_1) \geq 2$.

Clearly, there is at most one faulty element in $G_2$ and $E_c$. By
Lemma \ref{lem:YANG}, there exists a hamiltonian cycle $C_1$ in $G_1
- F_1$.

\bigskip \noindent \emph{Case 2.1.} $s, t \in V(G_1)$.

Clearly, $s$ and $t$ are both on $C_1$.

\bigskip \noindent \emph{Case 2.1.1.} $dist_{C_1}(s, t)= 1$ (see Fig. \ref{fig:2.1.1}).

We can find an edge $(u_1, v_1)$ on $C_1$ such that  $u_2 \in
N_{G-F}(u_1)$, $v_2 \in N_{G-F}(v_1)$, and $(u_1, v_1) \neq (s,t)$.
We may write $C_1$ as $\langle s, P_{11},u_1,v_1,P_{12},t,s\rangle$.
By Lemma \ref{lem:hamilton}, there exists a hamiltonian path $P_2$
between $u_2$ and $v_2$ in $G_2 - F_2$. Thus, $\langle
s,P_{11},u_1,u_2,P_2,v_2,v_1,P_{12},t\rangle$ forms a hamiltonian
path between $s$ and $t$ in $G - F$.

\begin{figure}
    \centering
    \includegraphics[scale=0.7]{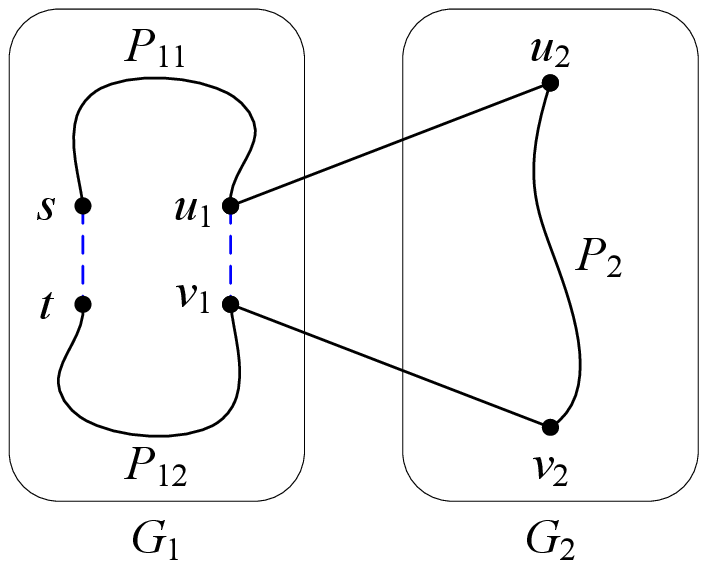}
    \caption{Illustration of the hamiltonian path in Case 2.1.1}
    \label{fig:2.1.1}
\end{figure}

\bigskip \noindent \emph{Case 2.1.2.} $dist_{C_1}(s, t)= 2$.

We may write $C_1$ as $\langle s, x_1,t,P_1,s\rangle$, then $P_1$ is
a near-hamiltonian path between $s$ and $t$ in $G_1-F_1$.

\bigskip \noindent \emph{Case 2.1.2.1.} $x_2 \in N_{G-F}(x_1)$ (see Fig.
\ref{fig:2.1.2.1}).

We may write $P_1$ as $\langle s,z_1,P_{11},y_1,t\rangle$. Because
there is at most one faulty element in $G_2$ and $E_c$, we can
choose one node out of $y_1$ and $z_1$, say $y_1$, such that $y_2
\in N_{G-F}(y_1)$. By Lemma \ref{lem:hamilton}, there exists a
hamiltonian path $P_2$ between $x_2$ and $y_2$ in $G_2 - F_2$. Thus,
$\langle s,z_1,P_{11},y_1,y_2,P_2,x_2,x_1,t\rangle$ forms a
hamiltonian path between $s$ and $t$ in $G - F$.

\begin{figure}
    \centering
    \includegraphics[scale=0.7]{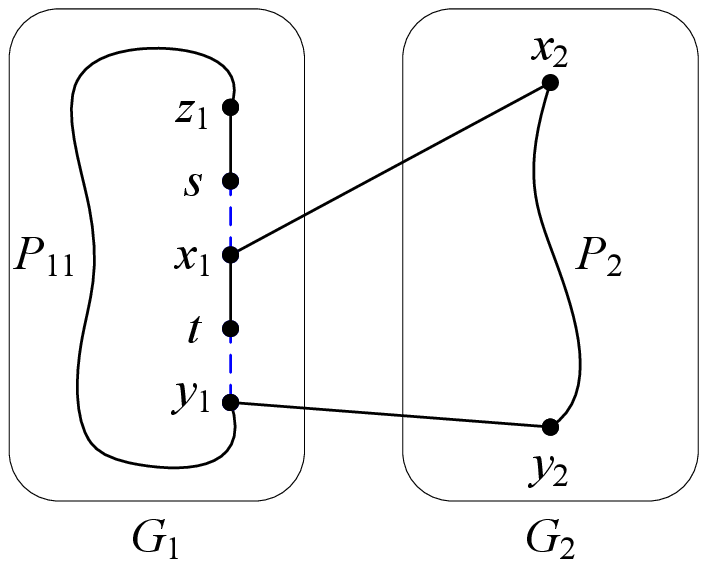}
    \caption{Illustration of the hamiltonian path in Case 2.1.2.1}
    \label{fig:2.1.2.1}
\end{figure}

\bigskip \noindent \emph{Case 2.1.2.2.} $x_2 \notin N_{G-F}(x_1)$ (see Fig.
\ref{fig:2.1.2.2}).

Since there is only one faulty element in $G_2$ and $E_c$ which
excludes $x_2$ from $N_{G-F}(x_1)$, then for an arbitrary chosen
edge $(u_1, v_1)$ on $P_1$, we have $u_2 \in N_{G-F}(u_1)$ and $v_2
\in N_{G-F}(v_1)$.

We may write $P_1$ as $\langle s,P_{12},u_1,v_1,P_{11},t\rangle$. By
Lemma \ref{lem:hamilton}, there exists a hamiltonian path $P_2$
between $u_2$ and $v_2$ in $G_2 - F_2$. Thus, $\langle
s,P_{12},u_1,u_2,$ $P_2,v_2,v_1,P_{11},t\rangle$ forms a
near-hamiltonian path between $s$ and $t$ in $G - F$.

\begin{figure}
    \centering
    \includegraphics[scale=0.7]{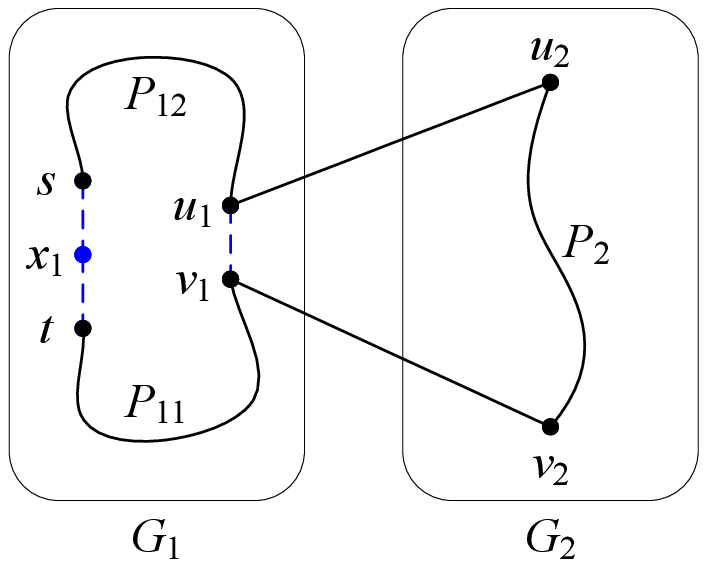}
    \caption{Illustration of the near-hamiltonian path in Case 2.1.2.2}
    \label{fig:2.1.2.2}
\end{figure}

\bigskip \noindent \emph{Case 2.1.3.} $dist_{C_1}(s, t)\geq 3$ (see Fig.
\ref{fig:2.1.3}).

We may write $C_1$ as $\langle x_1,s,u_1,P_{11},y_1,t,v_1,
P_{12},x_1\rangle$. Since there is at most one faulty element in
$G_2$ and $E_c$, then we can choose one out of two pairs of nodes
$[x_1,y_1]$ and $[u_1,v_1]$, say $[u_1,v_1]$, such that $u_2 \in
N_{G-F}(u_1)$ and $v_2 \in N_{G-F}(v_1)$.

By Lemma \ref{lem:hamilton}, there exists a hamiltonian path $P_2$
between $u_2$ and $v_2$ in $G_2 - F_2$. Thus, $\langle
s,x_1,P_{12},v_1,v_2,P_2,u_2,u_1,P_{11},y_1,t\rangle$ forms a
hamiltonian path between $s$ and $t$ in $G - F$.

\begin{figure}
    \centering
    \includegraphics[scale=0.7]{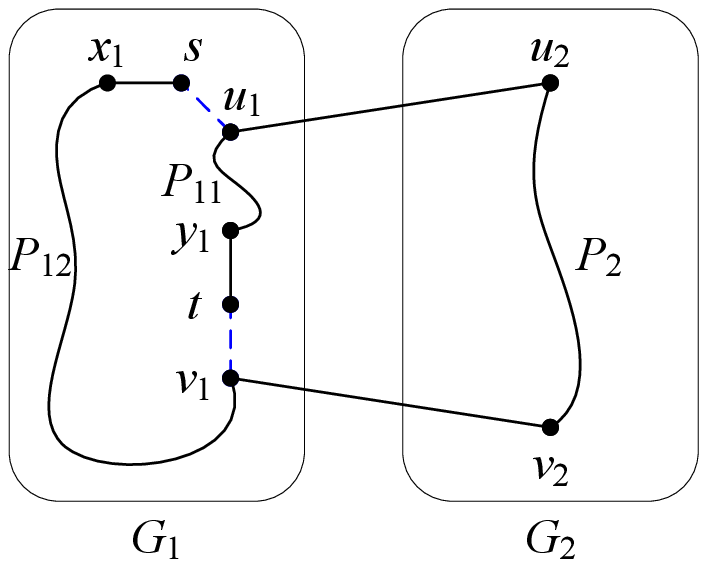}
    \caption{Illustration of the hamiltonian path in Case 2.1.3}
    \label{fig:2.1.3}
\end{figure}

\bigskip \noindent \emph{Case 2.2.} $s, t \in V(G_2) $ (see Fig. \ref{fig:2.2}).

We can find an edge $(u_1, v_1)$ on $C_1$ such that $u_2 \in
N_{G-F}(u_1)$, $v_2 \in N_{G-F}(v_1)$ and $\{ u_2, v_2\} \bigcap
\{s, t\} = \emptyset$. We may write $C_1$ as $\langle
u_1,P_1,v_1,u_1\rangle$. By Lemma \ref{lem:disjoint}, there exist
two paths $P_{21}$ and $P_{22}$ in $G_2 - F_2$ such that $P_{21}$
connects $s$ and $u_2$, $P_{22}$ connects $v_2$ and $t$, $V(P_{21})
\bigcap V(P_{22}) = \emptyset$, and $V(P_{21}) \bigcup V(P_{22} ) =
V(G_2 - F_2) $. Thus, $\langle
s,P_{21},u_2,u_1,P_1,v_1,v_2,P_{22},t\rangle$ forms a hamiltonian
path between $s$ and $t$ in $G - F$.

\begin{figure}
    \centering
    \includegraphics[scale=0.7]{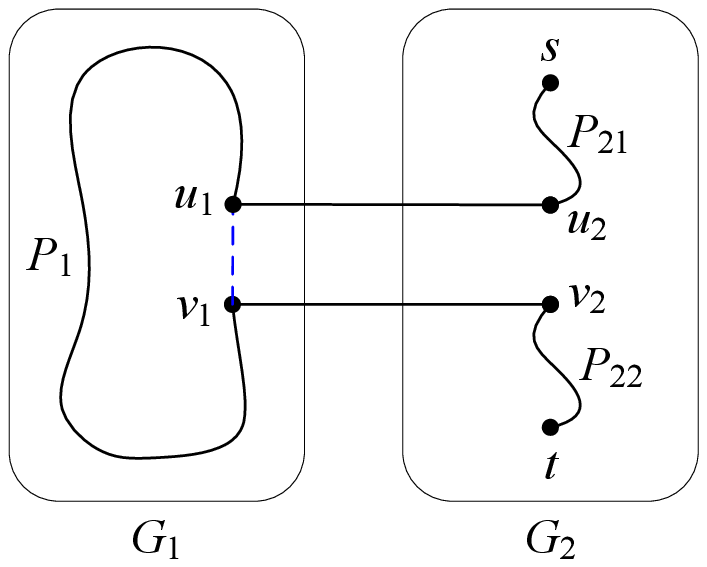}
    \caption{Illustration of the hamiltonian path in Case 2.2}
    \label{fig:2.2}
\end{figure}

\bigskip \noindent \emph{Case 2.3.} $s\in V(G_1)$ and $t \in V(G_2)$, or, $s\in V(G_2)$ and $t \in V(G_1)$.

W.L.O.G., we may assume $s\in V(G_1)$ and $t \in V(G_2)$. Clearly,
$s$ is on $C_1$. We may write $C_1$ as $\langle
u_1,s,v_1,P_1,u_1\rangle$.

\bigskip \noindent \emph{Case 2.3.1.} $u_2\in N_{G-F}(u_1)$ and $u_2 \neq t$, or, $v_2\in N_{G-F}(v_1)$ and $v_2 \neq t$.

W.L.O.G., we may assume $v_2\in N_{G-F}(u_1)$ and $v_2 \neq t$ (see
Fig. \ref{fig:2.3.1}). By Lemma \ref{lem:hamilton}, there exists a
hamiltonian path $P_2$ between $v_2$ and $t$ in $G_2 - F_2$. Thus,
$\langle s, u_1,P_1,v_1,v_2,P_2,t\rangle$ forms a hamiltonian path
between $s$ and $t$ in $G - F$.

\begin{figure}
    \centering
    \includegraphics[scale=0.7]{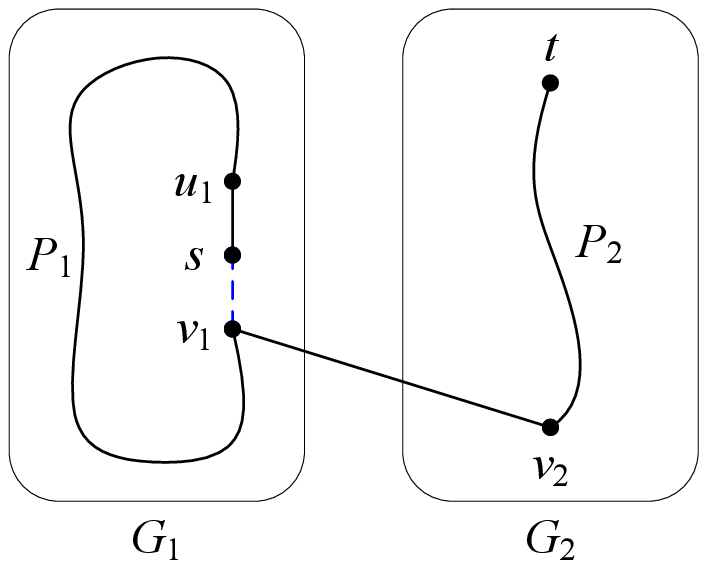}
    \caption{Illustration of the hamiltonian path in Case 2.3.1}
    \label{fig:2.3.1}
\end{figure}

\bigskip \noindent \emph{Case 2.3.2.} $u_2\notin N_{G-F}(u_1)$ and
$v_2 = t$, or, $v_2\notin N_{G-F}(v_1)$ and $u_2 = t$.

W.L.O.G., we may assume $u_2\notin N_{G-F}(u_1)$ and $v_2 = t$ (see
Fig. \ref{fig:2.3.2}). We can find an edge $(x_1,y_1)$ on $P_1$ such
that $\{x_1,y_1\}\bigcap \{u_1,v_1\} = \emptyset$. Since there is
only one faulty element in $G_2$ and $E_c$ which excludes $u_2$ from
$N_{G-F}(u_1)$, we have $x_2 \in N_{G-F}(x_1)$ and $y_2 \in
N_{G-F}(y_1)$. We may write $P_1$ as $\langle
u_1,P_{11},y_1,x_1,P_{12},v_1\rangle$. By Lemma \ref{lem:hamilton},
there exists a hamiltonian path $P_2$ between $x_2$ and $y_2$ in
$G_2 - F_2 - \{t\}$. Thus, $\langle
s,u_1,P_{11},y_1,y_2,P_2,x_2,x_1,P_{12},v_1,t\rangle$ forms a
hamiltonian path between $s$ and $t$ in $G - F$.

\begin{figure}
    \centering
    \includegraphics[scale=0.7]{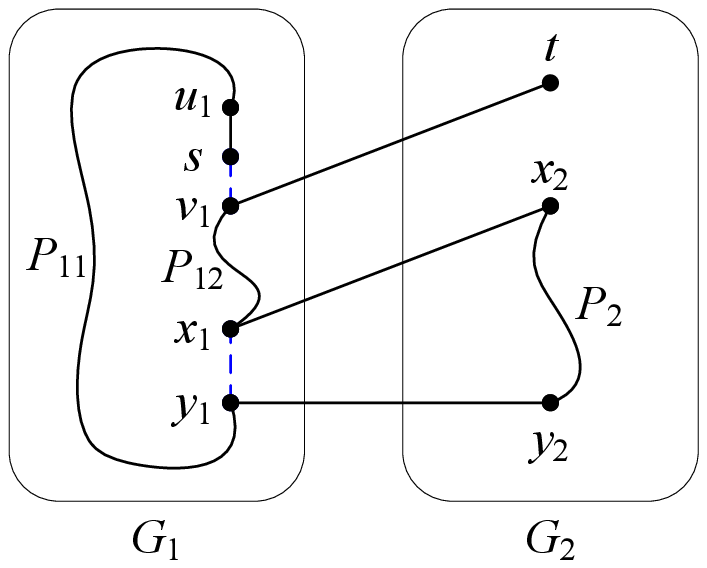}
    \caption{Illustration of the hamiltonian path in Case 2.3.2}
    \label{fig:2.3.2}
\end{figure}

\bigskip \noindent \emph{Case 3.} $\vert F_{1}\vert  =  2k - 9$ and $\delta (G_1 - F_1) \leq 1$.

Since $\delta (G_1 - F_1) \leq 1$, there exists a node, say $q_1\in
V(G_1)$, such that $\deg_{G_1-F_1}(q_1)\leq 1$. Since $\vert
F_{1}\vert  =  2k - 9$ and $\vert F\vert  \leq  2k - 8$, there is at
most one faulty element in $G_2$ and $E_c$. According to Lemma
\ref{lem:YANG}, there exists a near-hamiltonian cycle $C_1$ in $G_1
- F_1$, where $q_1$ is not on $C_1$.

\bigskip \noindent \emph{Case 3.1.} $s, t \in V(G_1)$.

\bigskip \noindent \emph{Case 3.1.1.} $s, t\neq q_1$.

Clearly, $s$ and $t$ are both on $C_1$.

\bigskip \noindent \emph{Case 3.1.1.1.} $dist_{C_1}(s,t) =1$ or $dist_{C_1}(s,t) \geq 3$.

The proof is similar to that of Cases 2.1.1 and 2.1.3.

\bigskip \noindent \emph{Case 3.1.1.2.} $dist_{C_1}(s, t)= 2$ (see Fig. \ref{fig:3.1.1.2}).

We may write $C_1$ as $\langle s, x_1,t,P_1,s\rangle$. If $x_2 \in
N_{G-F}(x_1)$, the proof is similar to that of Case 2.1.2.1. Here we
assume $x_2 \notin N_{G-F}(x_1)$.

Since $|F_1|=2k-9$ and $\deg_{G_1-F_1}(q_1)\leq 1$, then
$\deg_{G_1-F_1}(x_1)\geq k-\big((2k-9)-(k-1)\big)=8$. Thus, we can
find a node $y_1\in V(G_1-F_1)$ such that $(x_1,y_1)\in E(G_1-F_1)$
and $y_1$ is on $C_1$. We may rewrite $C_1$ as $\langle
s,x_1,t,P_{11},u_1,y_1,P_{12},v_1,s\rangle$. Since there is only one
faulty element in $G_2$ and $E_c$ which excludes $x_2$ from
$N_{G-F}(x_1)$, we have $u_2 \in N_{G-F}(u_1)$ and $v_2 \in
N_{G-F}(v_1)$.

By Lemma \ref{lem:hamilton}, there exists a hamiltonian path $P_2$
between $u_2$ and $v_2$ in $G_2 - F_2$. Thus, $\langle
s,x_1,y_1,P_{12},v_1,v_2,P_2,u_2,u_1,P_{11},t\rangle$ forms a
near-hamiltonian path between $s$ and $t$ in $G - F$.

\begin{figure}
    \centering
    \includegraphics[scale=0.7]{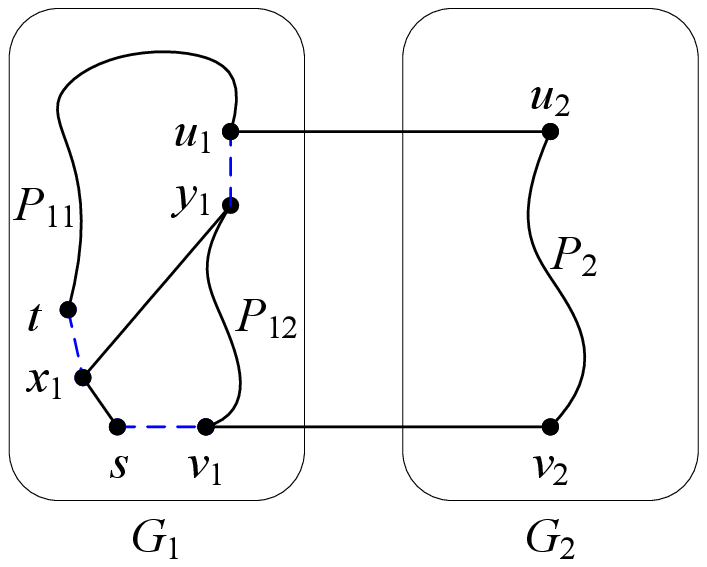}
    \caption{Illustration of the near-hamiltonian path in Case 3.1.1.2}
    \label{fig:3.1.1.2}
\end{figure}

\bigskip \noindent \emph{Case 3.1.2.} $s = q_1$ or $t = q_1$.

W.L.O.G., we may assume that $t = q_1$. If $q_2 \in N_{G-F}(q_1)$,
then we can regard $q_2$ as the agent of $t$ in $G_2$. Thus, the
proof is similar to that of Case 2.3.1. Here we assume that $q_2
\notin N_{G-F}(q_1)$. According to the assumption
$\deg_{G-F-\{s\}}(t) \geq 1$, we have $\deg_{G_1-F_1-\{s\}}(t) \geq
1$. Then we can find a node $t'\in V\big(G_1-F_1-\{s\}\big)$ such
that $(t,t')\in E\big(G_1-F_1-\{s\}\big)$. Clearly, $t'$ is on $C_1$
and $t' \neq s$. We can regard $t'$ as the agent of $t$ on $C_1$,
then the rest of the proof is similar to that of Case 3.1.1.

\bigskip \noindent \emph{Case 3.2.} $s, t \in V(G_2) $.

The proof is similar to that of Case 2.2.

\bigskip \noindent \emph{Case 3.3.} $s\in V(G_1)$ and $t \in V(G_2)$, or, $s\in V(G_2)$ and $t \in V(G_1)$.

W.L.O.G., we may assume $s\in V(G_1)$ and $t \in V(G_2)$.

\bigskip \noindent \emph{Case 3.3.1.} $s \neq q_1$.

The proof is similar to that of Case 2.3.

\bigskip \noindent \emph{Case 3.3.2.} $s = q_1$.

If $q_2\in N_{G-F}(q_1)$, then the proof is similar to that of Case
2.2. Here we assume that $q_2 \notin N_{G-F}(q_1)$. According to the
assumption $\deg_{G-F-\{t\}}(s) \geq 1$, we have $\deg_{G_1-F_1}(s)
\geq 1$. Then we can find a node $s'\in V(G_1-F_1)$ such that
$(s,s')\in E(G_1-F_1)$. Clearly, $s'$ is on $C_1$. We can regard
$s'$ as the agent of $s$ on $C_1$, then the rest of the proof is
similar to that of Case 2.3.

\bigskip \noindent \emph{Case 4.} $\vert F_{1}\vert  =  2k - 8$ and $\delta (G_1 - F_1) \geq 2$.

Clearly, there exists no faulty element in $G_2$ and $E_c$. Imagine
an arbitrarily chosen faulty element $f_e \in F_1$ to be fault-free.
By Lemma \ref{lem:YANG}, there exists a hamiltonian cycle $C_1$ in
$G_1 - F_1 + \{f_e\}$. We may write $C_1$ as $\langle u_1, f_e, v_1,
P_1, u_1\rangle$, then $P_1$ is a hamiltonian path in $G_1 - F_1$.

\bigskip \noindent \emph{Case 4.1.} $s, t \in V(G_1 )$.

Clearly, $s$ and $t$ are both on $P_1$. W.L.O.G., we may assume that
$s$ is a $u_1$-closer node on $P_1$, $t$ is a $v_1$-closer node on
$P_1$, and $d_{P_1}(s, u_1) \leq d_{P_1}(t, v_1)$.

\bigskip \noindent \emph{Case 4.1.1.} $d_{P_1}(s, t) = 1$ (see Fig. \ref{fig:4.1.1}).

We may write $P_1$ as $\langle u_1,P_{11},s,t,P_{12},v_1\rangle$. By
Lemma \ref{lem:hamilton}, there exists a hamiltonian path $P_2$
between $u_2$ and $v_2$ in $G_2$. Thus, $\langle
s,P_{11},u_1,u_2,P_2,v_2,v_1,P_{12},t\rangle$ forms a hamiltonian
path between $s$ and $t$ in $G - F$.

\begin{figure}
    \centering
    \includegraphics[scale=0.7]{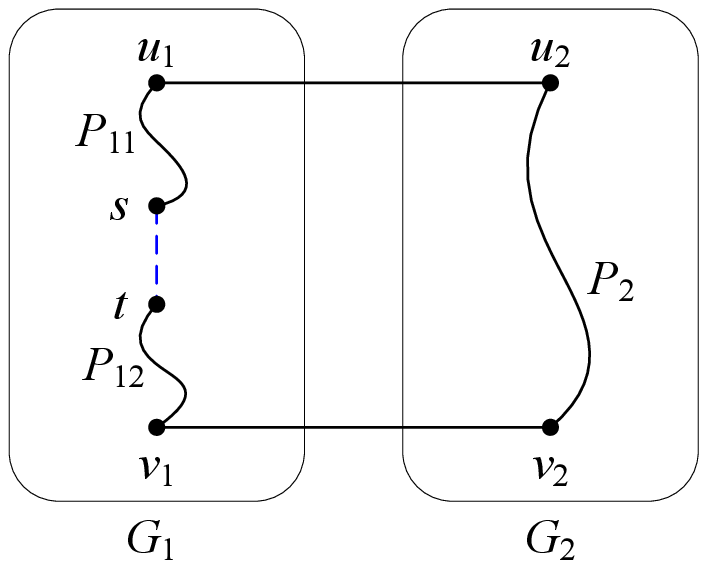}
    \caption{Illustration of the hamiltonian path in Case 4.1.1}
    \label{fig:4.1.1}
\end{figure}

\bigskip \noindent \emph{Case 4.1.2.} $d_{P_1}(s, t) = 2$ (see Fig. \ref{fig:4.1.2}).

we may write $P_1$ as $\langle
u_1,P_{11},s,x_1,t,y_1,P_{12},v_1\rangle$. By Lemma
\ref{lem:disjoint}, there exist two paths $P_{21}$ and $P_{22}$ in
$G_2$ such that $P_{21}$ connects $u_2$ and $v_2$, $P_{22}$ connects
$x_2$ and $y_2$, $V(P_{21}) \bigcap V(P_{22} ) = \emptyset$, and
$V(P_{21}) \bigcup V(P_{22}) = V(G_2)$. Thus, $\langle
s,P_{11},u_1,u_2,P_{21},v_2,v_1,P_{12},y_1,y_2,$
$P_{22},x_2,x_1,t\rangle$ forms a hamiltonian path between $s$ and
$t$ in $G - F$.

\begin{figure}
    \centering
    \includegraphics[scale=0.7]{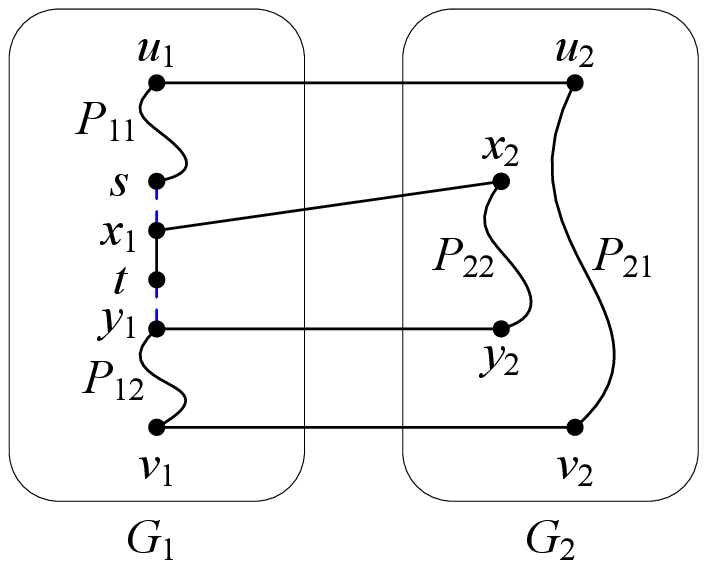}
    \caption{Illustration of the hamiltonian path in Case 4.1.2}
    \label{fig:4.1.2}
\end{figure}

\bigskip \noindent \emph{Case 4.1.3.} $d_{P_1}(s, t) \geq 3$ (see Fig. \ref{fig:4.1.3}).

we may write $P_1$ as $\langle
u_1,P_{11},s,x_1,P_{12},y_1,t,P_{13},v_1\rangle$. By Lemma
\ref{lem:disjoint}, there exist two paths $P_{21}$ and $P_{22}$ in
$G_2$ such that $P_{21}$ connects $u_2$ and $x_2$, $P_{22}$ connects
$v_2$ and $y_2$, $V (P_{21}) \bigcap V (P_{22} )=\emptyset$, and $V
(P_{21}) \bigcup V (P_{22} ) = V(G_2)$. Thus, $\langle
s,P_{11},u_1,u_2,P_{21},x_2,x_1,P_{12},$
$y_1,y_2,P_{22},v_2,v_1,P_{13},t\rangle$ forms a hamiltonian path
between $s$ and $t$ in $G - F$.

\begin{figure}
    \centering
    \includegraphics[scale=0.7]{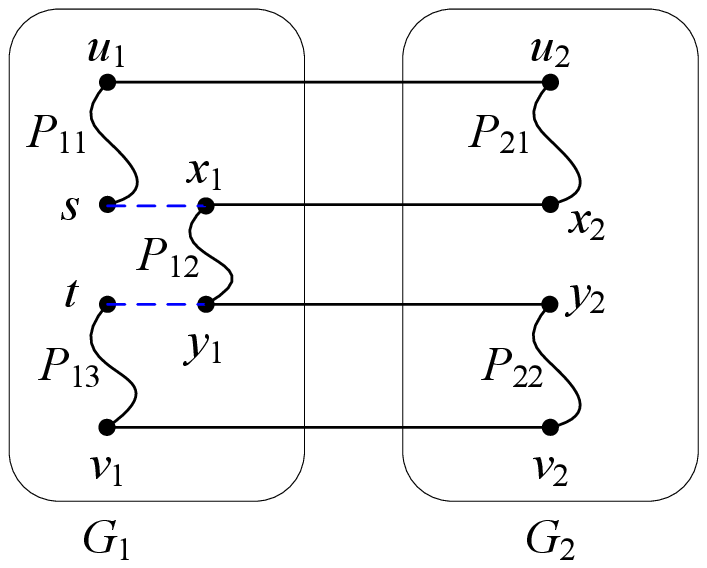}
    \caption{Illustration of the hamiltonian path in Case 4.1.3}
    \label{fig:4.1.3}
\end{figure}

\bigskip \noindent \emph{Case 4.2.} $s, t \in V(G_2 )$ (see Fig. \ref{fig:4.2.1}).

\bigskip \noindent \emph{Case 4.2.1.} $\{u_2, v_2\}\bigcap \{s, t\} = \emptyset$.

By Lemma \ref{lem:disjoint}, there exist two paths $P_{21}$ and
$P_{22}$ in $G_2$ such that $P_{21}$ connects $u_2$ and $s$,
$P_{22}$ connects $v_2$ and $t$, $V(P_{21}) \bigcap
V(P_{22})=\emptyset$, and $V(P_{21}) \bigcup V(P_{22}) = V(G_2)$.
Thus, $\langle s,P_{21},u_2,u_1,P_1,$ $v_1,v_2,P_{22},t\rangle$
forms a hamiltonian path between $s$ and $t$ in $G - F$.

\begin{figure}
    \centering
    \includegraphics[scale=0.7]{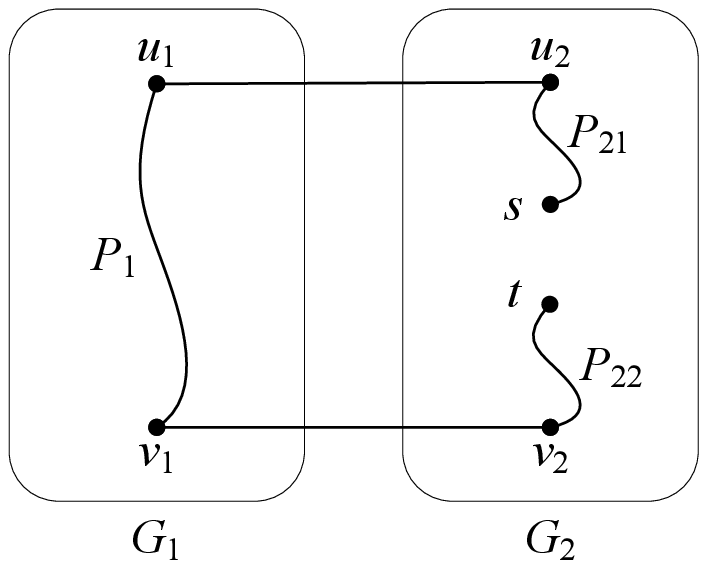}
    \caption{Illustration of the hamiltonian path in Case 4.2.1}
    \label{fig:4.2.1}
\end{figure}

\bigskip \noindent \emph{Case 4.2.2.} $\big|\{u_2, v_2\}\bigcap \{s, t\}\big| = 1$
(see Fig. \ref{fig:4.2.2}).

W.L.O.G., we may assume $u_2 = s$ and $v_2 \neq t$.  By Lemma
\ref{lem:hamilton}, there exists a hamiltonian path $P_2$ between
$v_2$ and $t$ in $G_2 - \{s\}$. Thus, $\langle
s,u_1,P_1,v_1,v_2,P_2,t\rangle$ forms a hamiltonian path between $s$
and $t$ in $G - F$.

\begin{figure}
    \centering
    \includegraphics[scale=0.7]{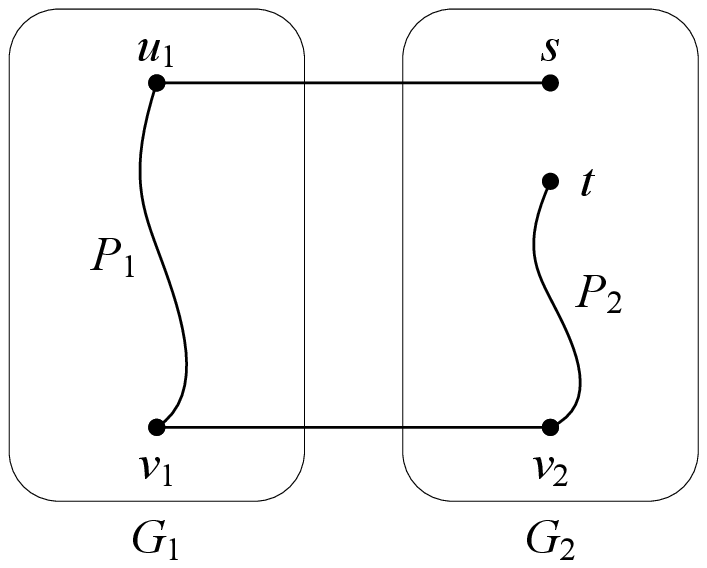}
    \caption{Illustration of the hamiltonian path in Case 4.2.2}
    \label{fig:4.2.2}
\end{figure}

\bigskip \noindent \emph{Case 4.2.3.} $\{u_2, v_2\}= \{s, t\}$.
(see Fig. \ref{fig:4.2.3})

W.L.O.G., we may assume $u_2 = s$ and $v_2 = t$. We can find an edge
$(x_1, y_1)$ on $P_1$ such that
$\{x_1,y_1\}\bigcap\{u_1,v_1\}=\emptyset$. We may write $P_1$ as
$\langle u_1,P_{11},x_1,y_1,P_{12},v_1\rangle$. By Lemma
\ref{lem:hamilton}, there exists a hamiltonian path $P_2$ between
$x_2$ and $y_2$ in $G_2 - \{s,t\}$. Thus, $\langle
s,u_1,P_{11},x_1,x_2,P_2,y_2,y_1,P_{12},v_1,t\rangle$ forms a
hamiltonian path between $s$ and $t$ in $G - F$.

\begin{figure}
    \centering
    \includegraphics[scale=0.7]{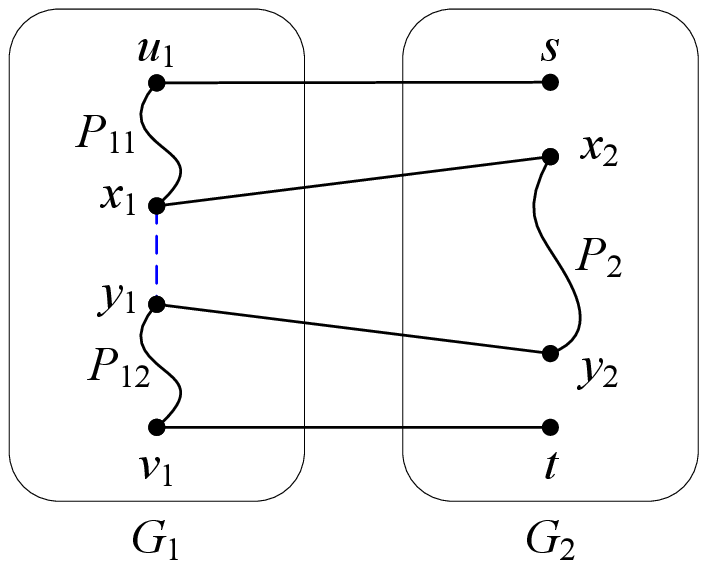}
    \caption{Illustration of the hamiltonian path in Case 4.2.3}
    \label{fig:4.2.3}
\end{figure}

\bigskip \noindent \emph{Case 4.3.} $s\in V(G_1)$ and $t \in V(G_2)$, or, $s\in V(G_2)$ and $t \in V(G_1)$.

W.L.O.G., we may assume $s\in V(G_1)$ and $t \in V(G_2)$. Clearly,
$s$ is on $P_1$. We may assume that $s$ is a $u_1$-closer node on
$P_1$. We may write $P_1$ as $\langle
u_1,P_{11},s,w_1,P_{12},v_1\rangle$.

\bigskip \noindent \emph{Case 4.3.1.} $u_2,w_2,v_2 \neq t$ (see Fig.
\ref{fig:4.3.1}).

By Lemma \ref{lem:disjoint}, there exist two paths $P_{21}$ and
$P_{22}$ in $G_2$ such that $P_{21}$ connects $u_2$ and $w_2$,
$P_{22}$ connects $v_2$ and $t$, $V(P_{21}) \bigcap
V(P_{22})=\emptyset$, and $V(P_{21}) \bigcup V(P_{22}) = V(G_2)$.
Thus, $\langle
s,P_{11},u_1,u_2,P_{21},w_2,w_1,P_{12},v_1,v_2,P_{22},t\rangle$
forms a hamiltonian path between $s$ and $t$ in $G - F$.

\begin{figure}
    \centering
    \includegraphics[scale=0.7]{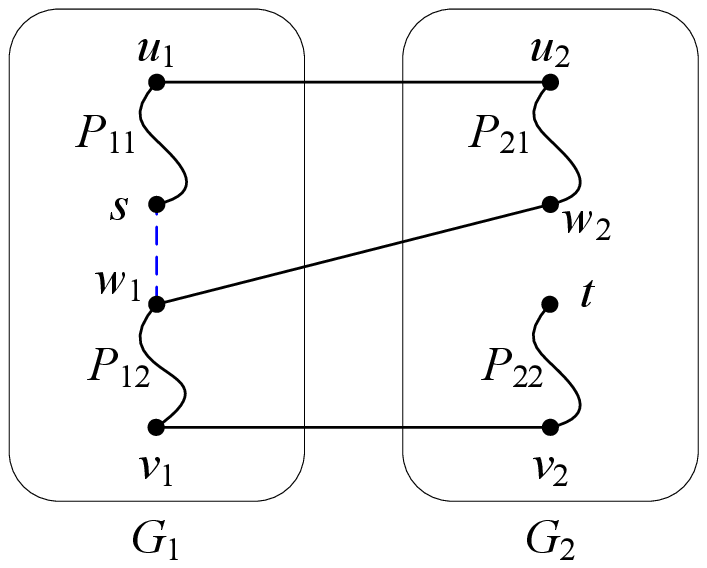}
    \caption{Illustration of the hamiltonian path in Case 4.3.1}
    \label{fig:4.3.1}
\end{figure}

\bigskip \noindent \emph{Case 4.3.2.} $v_2 = t$ or $w_2 = t$ (see Fig.
\ref{fig:4.3.2}).

W.L.O.G., we may assume $v_2 = t$. By Lemma \ref{lem:hamilton},
there exists a hamiltonian path $P_2$ between $u_2$ and $w_2$ in
$G_2 - \{t\}$. Thus, $\langle
s,P_{11},u_1,u_2,P_2,w_2,w_1,P_{12},v_1,t\rangle$ forms a
hamiltonian path between $s$ and $t$ in $G - F$.

\begin{figure}
    \centering
    \includegraphics[scale=0.7]{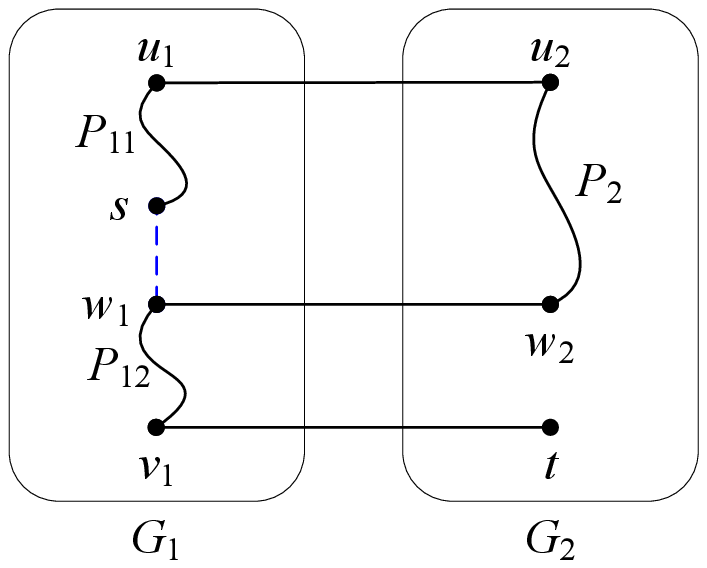}
    \caption{Illustration of the hamiltonian path in Case 4.3.2}
    \label{fig:4.3.2}
\end{figure}

\bigskip \noindent \emph{Case 4.3.3.} $u_2 = t$.

\bigskip \noindent \emph{Case 4.3.3.1.} $dist_{P_1}(u_1,s) = 1$ (see Fig.
\ref{fig:4.3.3.1}).

Since $\delta(G_1-F_1)\geq 2$, we can find a node $x_1$ on $P_{12}$
such that $(u_1, x_1) \in E(G_1-F_1)$. W.L.O.G., we may assume that
$x_1$ is a $v_1$-closer node on $P_{12}$. We may write $P_{12}$ as
$\langle w_1,P_{13},y_1,x_1,P_{14},v_1\rangle$.

By Lemma \ref{lem:disjoint}, there exist two paths $P_{21}$ and
$P_{22}$ in $G_2$ such that $P_{21}$ connects $w_2$ and $t$,
$P_{22}$ connects $v_2$ and $y_2$, $V(P_{21}) \bigcap
V(P_{22})=\emptyset$, and $V(P_{21}) \bigcup V(P_{22}) = V(G_2)$.
Thus, $\langle
s,u_1,x_1,P_{14},v_1,v_2,P_{22},y_2,y_1,P_{13},w_1,w_2,P_{21},t\rangle$
forms a hamiltonian path between $s$ and $t$ in $G - F$.

\begin{figure}
    \centering
    \includegraphics[scale=0.7]{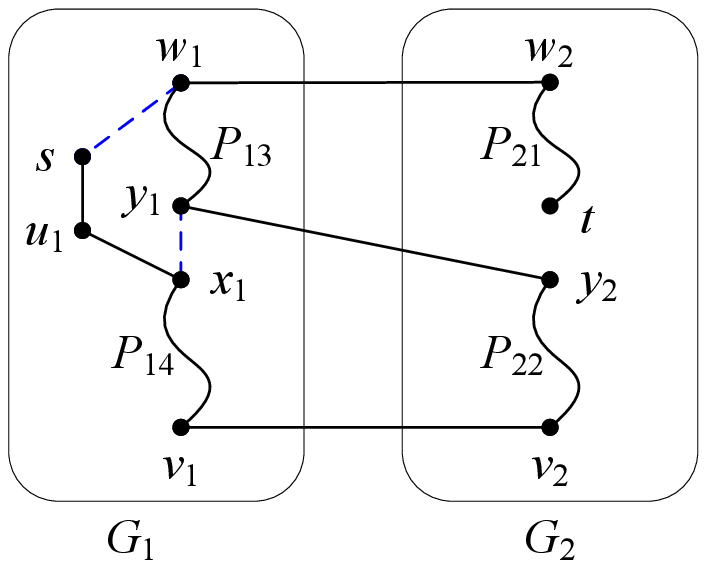}
    \caption{Illustration of the hamiltonian path in Case 4.3.3.1}
    \label{fig:4.3.3.1}
\end{figure}

\bigskip \noindent \emph{Case 4.3.3.2.} $dist_{P_1}(u_1,s) \geq 2$ (see Fig. \ref{fig:4.3.3.2}).

we may write $P_{11}$ as $\langle u_1,P_{13},x_1,s\rangle$. Let $s_1
= s$. By Lemma \ref{lem:disjoint}, there exist two paths $P_{21}$
and $P_{22}$ in $G_2$ such that $P_{21}$ connects $s_2$ and $w_2$,
$P_{22}$ connects $v_2$ and $x_2$, $V(P_{21}) \bigcap
V(P_{22})=\emptyset$, and $V(P_{21}) \bigcup V(P_{22}) =
V(G_2)-\{t\}$. Thus, $\langle s,s_2,P_{21},w_2,w_1,P_{12},v_1,v_2,$
$P_{22},x_2,x_1,P_{13},u_1,t\rangle$ forms a hamiltonian path
between $s$ and $t$ in $G - F$.

\begin{figure}
    \centering
    \includegraphics[scale=0.7]{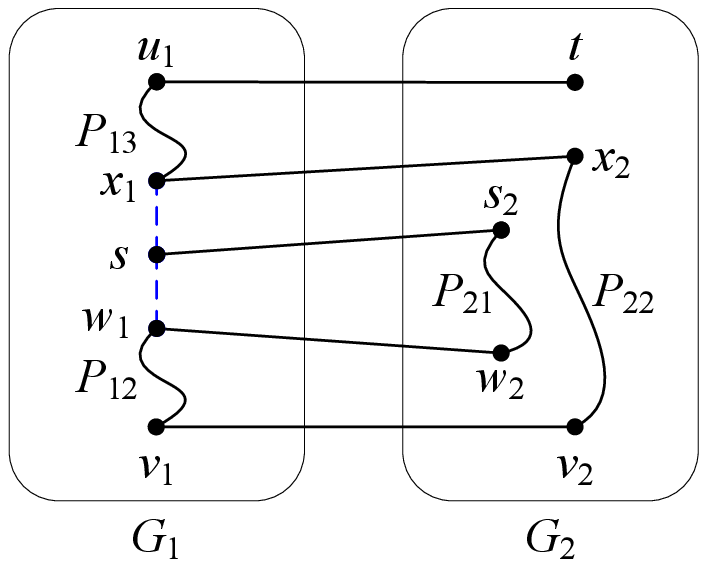}
    \caption{Illustration of the hamiltonian path in Case 4.3.3.2}
    \label{fig:4.3.3.2}
\end{figure}

\bigskip \noindent \emph{Case 5.} $\vert F_{1}\vert  =  2k - 8$ and $\delta (G_1 - F_1) \leq 1$.

Since $\delta (G_1 - F_1) \leq 1$, there exists a node, say $q_1\in
V(G_1)$, such that $\deg_{G_1-F_1}(q_1)\leq 1$. Since $\vert
F_{1}\vert  =  \vert F\vert  =  2k - 8$, there exists no faulty
element in $G_2$ and $E_c$. Imagine an arbitrarily chosen faulty
element $f_e \in F_1$ to be fault-free. Then by Lemma
\ref{lem:YANG}, there exists a near-hamiltonian cycle $C_1$ in $G_1
- F_1 + \{f_e\}$. We may write $C_1$ as $\langle u_1, f_e, v_1, P_1,
u_1\rangle$, then $P_1$ is a near-hamiltonian path in $G_1 - F_1$,
where $q_1$ is not on $P_1$.

\bigskip \noindent \emph{Case 5.1.} $s, t \in V(G_1 ) $.

\bigskip \noindent \emph{Case 5.1.1.} $s, t\neq q_1$.

The proof is similar to that of Case 4.1.

\bigskip \noindent \emph{Case 5.1.2.} $s = q_1$ or $t = q_1$.

W.L.O.G., we may assume that $t = q_1$. If $q_2 \in N_{G-F}(q_1)$,
then we can regard $q_2$ as the agent of $t$ in $G_2$. Thus, the
proof is similar to that of Case 4.3.1. Here we assume that $q_2
\notin N_{G-F}(q_1)$. According to the assumption
$\deg_{G-F-\{s\}}(t) \geq 1$, we have $\deg_{G_1-F_1-\{s\}}(t) \geq
1$. We can find a node $t'\in V\big(G_1-F_1-\{s\}\big)$ such that
$(t,t')\in E\big(G_1-F_1-\{s\}\big)$. Clearly, $t'$ is on $P_1$ and
$t' \neq s$. We can regard $t'$ as the agent of $t$ on $P_1$, then
the rest of the proof is similar to that of Case 4.1.

\bigskip \noindent \emph{Case 5.2.} $s, t \in V(G_2 )$.

The proof is similar to that of Case 4.2.

\bigskip \noindent \emph{Case 5.3.}  $s\in V(G_1)$ and $t \in V(G_2)$, or, $s\in V(G_2)$ and $t \in V(G_1)$.

W.L.O.G., we may assume $s\in V(G_1)$ and $t \in V(G_2)$.

\bigskip \noindent \emph{Case 5.3.1.} $s \neq q_1$.

W.L.O.G., we may assume that $s$ is a $u_1$-closer node on $P_1$. We
may write $P_1$ as $\langle u_1,P_{11},s,w_1,$ $P_{12},v_1\rangle$.
If $u_2 \neq t$, the proof is similar to that of Case 4.3.1 and
4.3.2. Here we assume that $u_2 = t$.

\bigskip \noindent \emph{Case 5.3.1.1.} $d_{P_1}(u_1,s) = 1$.

Since $\deg_{G_1 - F_1}(q_1) \leq 1$ and $|F_1|=2k-8$, then
$\deg_{G_1-F_1}(u_1)\geq k-\big((2k-8)-(k-1)\big)=7$. Thus, we can
find a node $x_1\in G_1 - F_1$ such that $(u_1,x_1)\in E(G_1-F_1)$
and $x_1$ is on $P_1$. The rest of the proof is similar to that of
Case 4.3.3.1.

\bigskip \noindent \emph{Case 5.3.1.2.} $d_{P_1}(u_1,s) \geq 2$.

The proof is similar to that of Case 4.3.3.2.

\bigskip \noindent \emph{Case 5.3.2.} $s = q_1$.

If $q_2\in N_{G-F}(q_1)$, then the proof is similar to that of Case
4.2. Here we assume that $q_2 \notin N_{G-F}(q_1)$. According to the
assumption $\deg_{G-F-\{t\}}(s) \geq 1$, we have $\deg_{G_1-F_1}(s)
\geq 1$. We can find a node $s'$ such that $(s,s')\in E(G_1-F_1)$
and $s'$ is on $P_1$. We can regard $s'$ as the agent of $s$ on
$P_1$, then the rest of the proof is similar to that of Case 5.3.1.
\end{proof}

\section{Conclusion}
\label{}

It is well known that many classical parallel algorithms possess a
linear array structured task graph. In order to implement a linear
array structured parallel algorithm efficiently on a specific
parallel computing system, it is essential to map the tasks owned by
the parallel algorithm to the nodes of the underlying
interconnection network so that any two tasks that are adjacent in
the linear array are mapped to two adjacent nodes of the network. If
the number of tasks in the linear array structured parallel
algorithm equals the number of nodes in the associated
interconnection network, it is desirable for this network to have a
hamiltonian path.

In this paper, we studied the fault-tolerant hamiltonian or
near-hamiltonian path in an $n$-D THLN $G$ ($n \geq 7$) with a set
$F$ of up to $2n - 10$ faulty elements. We proved that for any two
nodes $s,t\in V(G-F)$ satisfying a necessary condition on neighbors
of $u$ and $v$, $G-F$ contains a hamiltonian or near-hamiltonian
path between $s$ and $t$. Consequently, a linear array structured
parallel algorithm can be efficiently implemented on a parallel
computing system with THLN as its interconnection network even with
faulty nodes and/or links. As a nontrivial extension of
\cite{YANG11}, our result extends further the fault-tolerant graph
embedding capability of THLNs.

In our opinion, the method developed in this paper is very powerful
for exploring the near-hamiltonian cycle and near-hamiltonian path
in other interconnection networks under the large fault model. The
embedding capability of paths and cycles of various lengths, i.e.,
the fault-tolerant pancyclicity and fault-tolerant panconnectivity
of THLNs under the large fault model remain yet to be solved. It is
also worthwhile to study how to embed meshes and tori into THLNs
under the large fault model.

\section*{Acknowledgement}

This work was supported by National Nature Science Foundation of
China (No.60973120, No.60903073, No.61003231, No.61103109,
No.11105025, and No.11105024), China Postdoctor Research Foundation
(No.20110491705), and Scientific Research Foundation for the
Returned Overseas Chinese Scholars, State Education Ministry
(No.GGRYJJ08-2).

\label{}

\begin{biography}[{\includegraphics[width=1in,height=1.25in,clip,keepaspectratio]{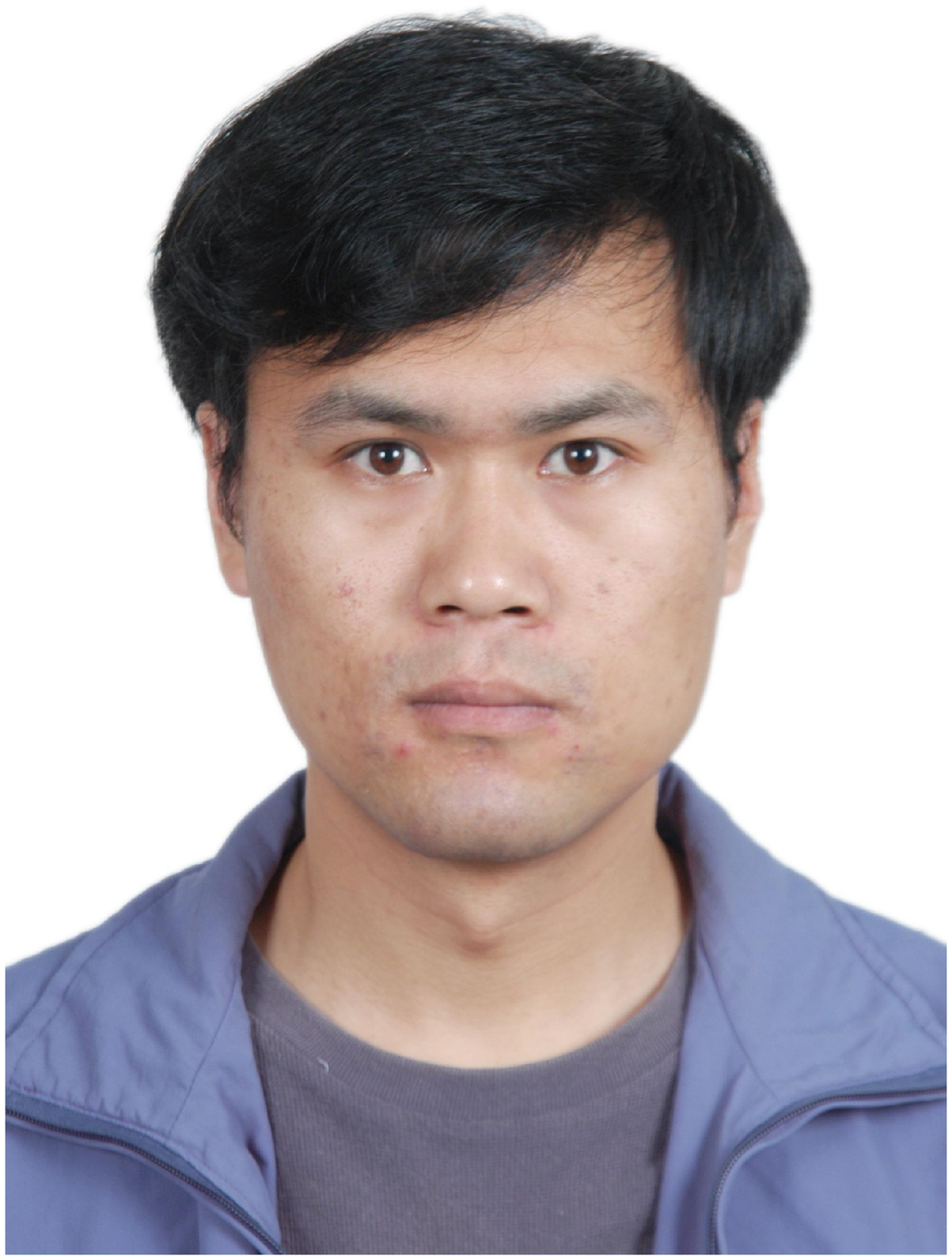}}]{Qiang Dong}
received the Ph.D. degree in Computer Science from Chongqing
University in 2010. He is now an assistant professor in the School
of Computer Science and Engineering, University of Electronic
Science and Technology of China, Chengdu, China. His current
research interests include interconnection networks,
Network-on-Chip, fault diagnosis, and complex networks.
\end{biography}







\begin{thebibliography}{1}

\bibitem{CULL95}
P. Cull, S.M. Larson, "The M\"obius cubes," \emph{IEEE Transactions
on Computers} \textbf{54(2)}, pp. 647-659, 1995.

\bibitem{DIESTEL}
R. Diestel, \emph{Graph Theory}, Springer-Verlag, New York, 2005.

\bibitem{DONG08}
Q. Dong, X. Yang, J. Zhao, Y.Y. Tang, "Embedding a family of
disjoint 3D meshes into a crossed cube," \emph{Information Sciences}
\textbf{178(11)} pp. 2396-2405, 2008.

\bibitem{DONG08IPL}
Q. Dong, X. Yang, J. Zhao, "Embedding a family of disjoint
multi-dimensional meshes into a crossed cube," \emph{Information
Processing Letters} \textbf{108(6)} pp. 394-397, 2008.

\bibitem{EFE92}
K. Efe, "The crossed cube architecture for parallel computation,"
\emph{IEEE Transactions on Parallel and Distributed Systetms}
\textbf{3(5)} pp. 513-524, 1992.

\bibitem{FAN11TCS}
J. Fan, X. Jia, B. Cheng, \emph{et al.}, "An efficient
fault-tolerant routing algorithm in bijective connection networks
with restricted faulty edges," \emph{Theoretcal Computer Science}
\textbf{412(29)} pp. 3440-3450, 2011.

\bibitem{FAN11}
J. Fan, X. Jia, X. Liu, \emph{et al.}, "Efficient unicast in
bijective connection networks with the restricted faulty node set,"
\emph{Information Sciences} \textbf{181(11)} pp. 2303-2315, 2011.

\bibitem{FAN10}
J. Fan, K. Li, S. Zhang, \emph{et al.}, "One-to-one communication in
twisted cubes under restricted connectivity," \emph{Frontiers of
Computer Science in China} \textbf{4(4)} pp. 489-499, 2010.

\bibitem{HF10}
Y. Han, J. Fan, S. Zhang, \emph{et al.}, "Embedding meshes into
locally twisted cubes," \emph{Information Sciences} \textbf{180(19)}
pp. 3794-3805, 2010.



\bibitem{HK87}
P.A.J. Hilbers, M.R.J. Koopman, J.L.A. van de Snepscheut, "The
twisted cube," \emph{Lecture Notes in Computer Science} \textbf{258}
(Springer, Berlin), Parallel Architectures and Languages Europe, pp.
152-159, 1987.

\bibitem{H04}
S.Y. Hsieh, C.H. Chen, "Pancyclicity on M\"obius cubes with maximal
edge faults," \emph{Parallel Computing} \textbf{30(3)} pp. 407-421,
2004.

\bibitem{H10}
S.Y. Hsieh, C.Y. Wu, "Edge-fault-tolerant hamiltonicity of locally
twisted cubes under conditional edge faults," \emph{Journal of
Combinatorial Optimization} \textbf{19(1)} pp. 16-30, 2010.

\bibitem{H10SIAM}
S.Y. Hsieh, C.W. Lee, "Pancyclicity of restricted hypercube-like
networks under the conditional fault model," \emph{SIAM Journal on
Discrete Mathematics} \textbf{23(4)} pp. 2010-2019, 2010.

\bibitem{LAI11}
P.L. Lai, "Geodesic pancyclicity of twisted cubes,"
\emph{Information Sciences}, \textbf{181(23)} pp. 5321-5332, 2011.

\bibitem{P99}
B. Parhami, \emph{An Introduction to Parallel Processing: Algorithms
and Architectures}, Plenum Press, New York, 1999.

\bibitem{PARK05}
J.H. Park, H.C. Kim, H.S. Lim, "Fault-Hamiltonicity of
hypercube-like interconnection networks," in: \emph{Proc. of IEEE
International Parallel \& Distributed Processing Symposium}, pp.
60a, 2005.

\bibitem{PARK09}
J.H. Park, H.C. Kim, H.S. Lim, "Many-to-many disjoint path covers in
the presence of faulty elements," \emph{IEEE Transactions on
Computers} \textbf{58(4)} pp. 528-540, 2009.

\bibitem{PARK07}
J.H. Park, H.S. Lim, H.C. Kim, "Panconnectivity and pancyclicity of
hypercube-like interconnection networks with faulty elements,"
\emph{Theoretical Computer Science} \textbf{377(1-3)} pp. 170-180,
2007.

\bibitem{WF11}
X. Wang, J. Fan, X. Jia, \emph{et al.}, "Embedding meshes into
twisted-cubes," \emph{Information Sciences} \textbf{181(14)} pp.
3085-3099, 2011.

\bibitem{XU11}
X. Xu, W. Zhai, J.M. Xu, \emph{et al.}, "Fault-tolerant
edge-pancyclic of locally twisted cubes," \emph{Information
Sciences} \textbf{181(11)} pp. 2268-2277, 2011.


\bibitem{XU06}
J.M. Xu, M. Ma, M. Lv, "Paths in M\"obius cubes and crossed cubes,"
\emph{Information Processing Letters} \textbf{97(3)} pp. 94-97,
2006.


\bibitem{YMC10}
M.C. Yang, "Constructing edge-disjoint spanning trees in twisted
cubes," \emph{Information Sciences} \textbf{180(20)} pp. 4075-4083,
2010.

\bibitem{YXF05}
X. Yang, D.J. Evans, G.M. Megson, "The locally twisted cubes,"
\emph{International Journal of Computer Mathematics} \textbf{82(4)}
pp. 401-413, 2005.

\bibitem{YXF10}
X. Yang, Q. Dong, Y.Y. Tang, "Embedding meshes/tori in faulty
crossed cubes," \emph{Information Processing Letters}
\textbf{110(14-15)} pp. 559-564, 2010.

\bibitem{YANG11}
X. Yang, Q. Dong, E. Yang, J. Cao, "Hamiltonian properties of
twisted hypercube-like networks with more faulty elements,"
\emph{Theoretical Computer Science} \textbf{412(22)} pp. 2409-2417,
2011.




\end{thebibliography}
\end{document}